\documentclass[conference,letterpaper]{IEEEtran}

\usepackage[utf8]{inputenc} 
\usepackage[T1]{fontenc}
\usepackage{url}
\usepackage{todonotes}
\usepackage{ifthen}
\usepackage{cite}
\usepackage[cmex10]{amsmath} 
\usepackage{multirow}
\usepackage{array}
\usepackage{hhline}
\usepackage{graphicx}
\usepackage{amsfonts}
\usepackage{amssymb}
\usepackage{amsthm}
\usepackage{xcolor}
\usepackage{algpseudocode}
\usepackage{algorithm}

\newtheorem{theorem}{Theorem}
\newtheorem{defi}{Definition}
\theoremstyle{definition}

\def\R{\mathbb{R}}
\def\dmat{X}  
\def\msg{M}
\newcommand{\expv}{\mathbb{E}}
\DeclareMathOperator*{\argmax}{argmax}
\DeclareMathOperator*{\argmin}{argmin}
\def\corr{r}
\def\nx{n_1}
\def\ny{n_2}

\def\t{Q}  
\def\k{T_0}  
\def\lowerdim{q}  
\newcommand{\indep}{\perp\!\!\!\perp}

\begin{document}
\title{Message-Relevant Dimension Reduction of Neural Populations} 


\author{%
  \IEEEauthorblockN{Amanda Merkley and Pulkit Grover}
  \IEEEauthorblockA{Department of Electrical \& Computer Engineering\\
                    Carnegie Mellon University, Pittsburgh, USA\\
                    Email: \{amerkley, pgrover\}@andrew.cmu.edu
                    }
  \and
  \IEEEauthorblockN{Alice Y.~Nam and Y.~Kate Hong}
  \IEEEauthorblockA{Department of Biological Sciences\\ 
                    Carnegie Mellon University, Pittsburgh, USA\\
                    Email: \{aynam, katehong\}@andrew.cmu.edu
                    }
}


\maketitle


\begin{abstract}
    Quantifying relevant interactions between neural populations is a prominent question in the analysis of high-dimensional neural recordings. However, existing dimension reduction methods often discuss communication in the absence of a formal framework, while frameworks proposed to address this gap are impractical in data analysis. This work bridges the formal framework of $\msg$-Information Flow with practical analysis of real neural data. To this end, we propose Iterative Regression, a message-dependent linear dimension reduction technique that iteratively finds an orthonormal basis such that each basis vector maximizes correlation between the projected data and the message. We then define `$\msg$-forwarding' to formally capture the notion of a message being forwarded from one neural population to another. We apply our methodology to recordings we collected from two neural populations in a simplified model of whisker-based sensory detection in mice, and show that the low-dimensional $M$-forwarding structure we infer supports biological evidence of a similar structure between the two original, high-dimensional populations.
\end{abstract}

\section{Introduction}
Noisy, high-dimensional data are collected in abundance in neuroscience experiments. It is now common to record tens to hundreds of neurons from multiple neural populations~\cite{steinmetz2018challenges, stevenson2011advances}, where each population is a functionally distinct group of neurons in the brain. These advances are enabling neuroscientists to ask questions about communication and information flow in real neural data to gain insight into population-level interactions. There is also a parallel interest to address these questions using information-theoretic formalisms~\cite{kim2014dynamic, schamberg2019measuring, quinn2017bounded}. A specific direction we took in earlier works~\cite{venkatesh2019should, venkatesh2020else, venkatesh2020information} was ``$M$-Information Flow'', a formal framework in which we defined information flow of $M$, a message of interest, and provided a technique to track the flow of $M$ through a network in time.

However, the rigorous definitions we laid out in~\cite{venkatesh2020information} are impractical in the analysis of noisy, high-dimensional neural datasets. To address this issue, we consider the problem of message-relevant dimension reduction from first principles. We seek an estimable representation of neural activity that summarizes information about the message, $\msg$, in neural data based on a measure of $\msg$-relevance. We propose \textit{Iterative Regression} (IR) (Sec.~\ref{sec:IR}) as a linear dimension reduction technique that uses sample correlation between $\msg$ and the projected data as the measure of $\msg$-relevance. IR iteratively finds an orthonormal basis to represent the neural data by computing the linear projection that maximizes correlation with the message.


Our second contribution is inference of a low-dimensional communication structure by conditional dependence testing using data reduced with IR. Inspired by $M$-Information Flow, we define `$\msg$-forwarding' in Sec.~\ref{sec:relay} as a conditional independence relationship to formally capture the notion of how a message $\msg$ is forwarded from one neural population to another. While $\msg$-forwarding generally describes a relationship between neural populations (i.e., vectors of activity), we estimate this quantity with one-dimensional (1D) representations of neural populations (i.e., scalars). Inferring an $\msg$-forwarding structure using IR-reduced data not only overcomes the challenge of directly applying the definition to high-dimensional datasets, but can also reveal the existence (or absence) of a low-dimensional representation of $\msg$-forwarding.

To assess the performance of our approach on real data, we analyze a simplified network of the whisker-based sensory detection system in mice. We first apply IR to neural population data we acquired and show that IR captures strong correlation with $\msg$ in each population (Sec.~\ref{sec:pca_ir_1pop}). We compare the correlations obtained from IR to those from demixed PCA (dPCA)~\cite{kobak2016demixed} and model-based Targeted Dimension Reduction (mTDR)~\cite{aoi2020prefrontal}, two techniques developed for message-dependent dimension reduction in neuroscience. We also compare to a denoising sparse autoencoder (DSAE). The DSAE of~\cite{qiu2018denoising} is a shallow autoencoder, suitable for our dataset size, used in neural signal classification. We adapt this model to see if nonlinear dimension reduction provides an advantage in spite of not explicitly accounting for $\msg$. We find that IR maintains clear $\msg$-relevance, while it is diminished in other methods (Sec.~\ref{sec:pca_ir_1pop}).

Next, we use conditional dependence testing to determine if 1D representations of both neural populations do \emph{not} harbor an $\msg$-forwarding structure (Sec.~\ref{sec:infer_m_comm}). These tests allow us to infer a direct pathway between $\msg$ and a dimension-reduced neural population, and that the message is not $\msg$-forwarded from the other reduced population. Using IR to obtain the 1D representation of each area allows us to conclude dependence of $\msg$ with both of the low-dimensional neural populations as well as a direct pathway between $\msg$ and one of the IR-reduced populations, suggesting a low-dimensional structure supported by previous anatomical evidence~\cite{hong2018sensation}. Furthermore, most other methods do not even detect unconditional dependence with one population. Only dPCA detects the same communication structure as IR, and to less confidence than IR.

\section{Linear Dimension Reduction}
\subsection{A measure of $\msg$-relevance}
We first address the problem of obtaining dimensions relevant to the message, $\msg$, from an $n$-dimensional dataset\footnote{In the rest of this work, we assume $\msg$ and $\dmat$ are mean-centered without loss of generality.}, $\dmat$. Naturally, the choice of dimension reduction depends on how relevance is measured. In this work, we measure $\msg$-relevance as sample correlation of the reduced-dimensional data with $\msg$ since correlation is both estimable and interpretable. As such, we seek projections of the data that retain the most correlation with $\msg$. The first projection is
\begin{align}\label{eq:max_corr}
    v_1 = \argmax_{u \in \R^n} r(u^T\dmat, \msg)
\end{align}
where $r(u^T\dmat, \msg)$ is the sample correlation between $u^T\dmat$ and $\msg$. In fact, we show in Appendix~\ref{app:a} that~\eqref{eq:max_corr} is equivalent to finding the linear minimum mean square error (LMMSE) estimate of $\msg$, to which the solution is ordinary least squares regression, i.e., $v_1 = (XX^T)^{-1}XM = \Sigma_X^{-1} \Sigma_{XM}$. When $X$ and $M$ are jointly Gaussian, the regression vector is also a sufficient transformation of $X$, which we show in Theorem \ref{thm:sufficiency} in Appendix~\ref{app:a}.

However, the regression vector may not be sufficient for representing neural data in general. As such, $\msg$ may not be recoverable by a single linear projection of $X$. A key observation is the existence of nonzero correlations in other dimensions. In other words, data projected onto $u \neq v_1$ may still be significantly correlated to $\msg$. Indeed, for $\expv[u^TXM] = 0$ to hold, the projection vector $u$ must be orthogonal to $\Sigma_{XM}$. This motivates IR as a method for recovering an orthonormal basis such that each basis vector iteratively maximizes $\msg$-relevance in decreasing dimensions after removing contributions from previous regression vectors. We consider linear projections to align with the dominant form of dimension reduction on neural data~\cite{cunningham2014dimensionality}. The focus on linear dimension reduction in neuroscience is partly due to the statistical versatility of linear techniques. Each dimension is also interpretable as a `conceptual neuron' that linearly combines real neurons to summarize some aspect of the data.

\subsection{Iterative Regression}\label{sec:IR}
We now describe IR in detail as a technique for extracting $\msg$-relevant projections one dimension at a time. Suppose $\dmat\in \R^{n \times \t \times \k}$ is the dataset of $n$ neurons recorded for $\k$ time points over $\t$ trials, where a trial is a time series of neural activity in response to a randomized stimulus. Let $\msg \in \R^\t$ be the message sampled over $\t$ trials. Define $\dmat_n \in \R^{n \times \t}$ to be a slice of the dataset at one time point. The first IR dimension, $v_1^*$, is the normalized regression vector from~\eqref{eq:max_corr}, i.e., $v_1^* = v_1 / \|v_1\|_2$. We assume $n < \t$, an assumption usually satisfied in experiments and is otherwise resolved by removing unresponsive neurons.

IR builds an orthonormal basis from the frame of reference of the first regression vector by finding correlations in a reduced $\lowerdim$-dimensional space after removing contributions from the previous $(n-\lowerdim)$ IR dimensions. To illustrate this process for the second IR dimension, consider the dataset after removing the contributions from $v_1^*$:
\begin{align*}
    \dmat_{n-1} = \dmat_n - v_1^*(v_1^*)^T\dmat_n.
\end{align*}
Whereas $\dmat_n$ was the dataset of real neural activity, $\dmat_{n-1}$ is now interpreted as the activity of conceptual neurons. By removing a column from $\dmat_{n-1}$ (the redundant conceptual neuron), the regression vector for the lower-dimensional space is again computed for the reduced set of conceptual neurons to find the next IR dimension $v_2 \in \R^{n-1}$. Then, $v_2$ is normalized, projected back into the $n$-dimensional space, and orthogonalized with respect to $v_1^*$. This process, described in Algorithm \ref{algo:ir}, is repeated for all $n$ IR dimensions at each time.
\begin{algorithm}
\caption{Iterative Regression}
\label{algo:ir}
\begin{algorithmic}[1]
    \Require $\dmat \in \R^{n \times \t \times \k}$, $\msg \in \R^\t$
    \State $M \gets M - \bar{M}$ \Comment{Mean center}
    \For{$t = 1, ..., \k$}
        \For{$i = 1, ..., n$}
            \State $\dmat_0 \gets \dmat(t) - \bar{\dmat}(t)$ \Comment{Mean center wrt $n$-axis}
            \For{$1 \leq j < i$}
                \State $s_j \gets v_j(t)^T\dmat_0$
                \State $s_j \gets s_j / \|s_j\|_2$
                \State $\dmat_0 \gets \dmat_0 - \dmat_0 s_j^T s_j$
            \EndFor
            \State $v_q \gets (\dmat_0\dmat_0^T)^{-1} \dmat_0 \msg$
            \State $v_{GS} \gets$ Gram-Schmidt($v_q, \{\varnothing, ..., v_{i-1}(t)\}$)
            \State $v_i(t) \gets v_{GS} \;/\; \|v_{GS}\|_2$
        \EndFor
    \EndFor
\end{algorithmic}
\end{algorithm}

Since neural spiking data is often collected as a $n\times \t \times \k$ tensor, the dataset must be converted into a matrix so it is amenable to dimension reduction with IR. A common technique for reshaping the tensor is to average over trials so the new dataset $X_{average} \in \R^{n\times \k}$ indexes neurons and time~\cite{chapin1999principal}. Dimension reduction is then performed on $X_{average}$. This operation effectively treats each time point as an independent sample of neural activity, when, by design of the experiment, activity within each trial is not even stationary. 

Instead, we apply dimension reduction per time point. This is to better respect the nonstationary nature of the trial, acknowledging that there may be dependence between time points. Furthermore, IR is a correlation-maximizing method applied to neurons that may be anti-correlated with $\msg$. Applying IR on the averaged dataset $X_{average}$ returns dimensions that maximize correlation over the whole trial, which may actually result in dimensions that are the most \textit{uncorrelated} representation of $\msg$. By performing dimension reduction per time point, the effect of anti-correlated neurons is accounted for simply by flipping signs at each time. 


\subsection{Related methods: Linear dimension reduction}
IR is related to several classes of linear dimension reduction techniques, including principal component analysis (PCA), canonical correlation analysis (CCA), and partial least squares (PLS) regression. To see this, consider the IR, PCA, CCA, and PLS objectives for finding the first dimension\footnote{We assume that $M \in \R$ and $X \in \R^n$. $\Sigma_X\in\R^{n \times n}$ is the covariance of $X$ and $\Sigma_{MX}\in\R^{1 \times n}$ is the cross covariance between $M$ and $X$.}:

\begin{itemize}
    \item IR objective
    \begin{align}
    \max_u \frac{\Sigma_{MX}u}{\sqrt{\sigma_m^2} \sqrt{u^T\Sigma_X u}}
    \end{align}
    \item PCA objective~\cite{bishop2006pattern}
    \begin{align}
        \max_u \frac{u^T\Sigma_X u}{u^Tu}
    \end{align}
    \item CCA objective~\cite{hardle2015canonical}
    \begin{align}
        \max_{v, u} \frac{v^T\Sigma_{MX}u}{\sqrt{v^T\sigma_m^2 v} \sqrt{u^T\Sigma_X u}} \label{obj:cca}
    \end{align}
    \item PLS objective~\cite{hoskuldsson1988pls}
    \begin{align}
        \max_{v, u} \frac{v^T\Sigma_{MX}u}{\sqrt{v^T v} \sqrt{u^T u}}
    \end{align}
\end{itemize}

We group the techniques by how they treat the message. IR and PLS are both message-dependent dimension reduction techniques, while PCA and CCA are not message-dependent. Although~\eqref{obj:cca} suggests that CCA is also message-dependent through $\Sigma_{MX}$, CCA is ill-suited for dimension reduction of $X$ when $M$ is 1D. Standard CCA was developed for dimension reduction of two homogeneous objects (e.g., both variables represent neural activity) unlike PLS, which was developed for robust regression~\cite{wold1984collinearity} (e.g., one variable is a stimulus and another is neural activity). The algorithmic description of CCA and PLS manifests this difference: the max number of dimensions standard CCA can find is limited to the minimum of dimensions of the two variables~\cite{bach2005probabilistic} (i.e. one dimension for a 1D message) while PLS is a sequential algorithm that squeezes out covariance after successive `deflations' of the data matrix $X$ (i.e. up to $n$ dimensions).

Although PCA is not message-dependent, $\msg$-relevance can still be assigned to each Principal Component (PC). In fact, PCA is related to IR by the projection of each PC onto the first IR dimension. For a given PC, the magnitude of this projection, scaled by the PC's explained variance, is proportional to its $\msg$-relevance. Theorem \ref{thm:pca_ir} summarizes this relationship.

\begin{theorem}\label{thm:pca_ir}
    Let $\dmat \in \R^{n \times \t}$ be the mean-centered dataset and $\msg\in \R^\t$, the mean-centered message vector. Suppose $v = (\dmat\dmat^T)^{-1}\dmat M$ is the regression vector and $\{p_i\}_{i=1}^n$ is the set of principal components of $\Sigma_X \propto \dmat\dmat^T$. Then,
    \begin{align}
        r(p_i^T\dmat, \msg) = \alpha \sqrt{\lambda_i} \ \langle v^*, p_i \rangle,
    \end{align}
    where $r(p_i^T\dmat, \msg)$ is the sample correlation coefficient between $p_i^T\dmat$ and $\msg$, $\lambda_i$ is the eigenvalue associated to $p_i$, $v^*$ is the first IR dimension, and $\alpha = \|v\| / s_\msg$ is a constant with $s_\msg$ being the standard deviation of $\msg$.
\end{theorem}
\begin{proof}
    See Appendix \ref{app:a}.
\end{proof}
PLS optimizes symmetrically over $\msg$ and $X$, where the additional optimization parameter $v$ of PLS yields a low-dimensional representation of $\msg$. However, since $\msg$ is a design variable of the experiment, we are only interested in the projection $u$ of the data $X$. In summary, IR can be viewed as: (1) a message-dependent formulation of PCA, (2) an asymmetric, sequential formulation of CCA, or (3) an asymmetric, correlation-based formulation of PLS. 


\subsection{Related methods: Dimension reduction in neuroscience}
dPCA~\cite{kobak2016demixed} and mTDR~\cite{aoi2018model} are respective  adaptations of PCA and regression for message-dependent dimension reduction in the context of neuroscience. dPCA `demixes' the contributions of different messages (e.g., experimental stimuli, animal behavioral features, and time indices) in the low-dimensional representation of neural activity~\cite{ kobak2016demixed}. dPCA achieves this separation by first decomposing the original neural data into components marginalized with respect to a specific message (or grouping of messages), then performs a low-rank matrix factorization on each component with the goal of minimizing overall reconstruction error~\cite{kobak2016demixed}, akin to regular PCA on the marginalized covariance. This differs from IR, which performs dimension reduction with $\msg$ directly, rather than a representation of $\msg$ derived from marginalized activity. 

mTDR is a regression-based method ~\cite{aoi2018model}. mTDR assumes a high-dimensional linear regression model of neural activity in response to multiple messages for each neuron and each time point, then performs a low-rank matrix factorization of the model. Unlike IR, mTDR finds all dimensions simultaneously. Although this may be advantageous for evaluating the entire low-dimensional model, the $\msg$-relevance of specific dimensions may be spread across multiple dimensions. Because IR is necessarily sequential, each dimension can be optimized for $\msg$-relevance. As shown in Sec. \ref{sec:infer_m_comm}, this may be crucial for inferring an accurate summarization of communication.

dPCA and mTDR are the primary message-dependent techniques in a large literature on dimension reduction in neuroscience~\cite{cunningham2014dimensionality}. While other techniques do not explicitly account for a message, they still discuss communication as low-dimensional interactions between neural populations by using various correlational measures~\cite{semedo2014extracting, cowley2017distance, gokcen2022disentangling}. Some of these works~\cite{gokcen2022disentangling} can be viewed as dimension-reduced extensions of Granger causality~\cite{ding2006granger} and directed information, which are also correlational measures of interaction, and have been of interest within information-theoretic literature on neuroscience~\cite{quinn2011estimating}. Nonlinear approaches in the form of autoencoders have also been explored, ranging from shallow DSAEs~\cite{qiu2018denoising} to deep, sequential networks~\cite{pandarinath2018inferring}. Although not message-dependent, these approaches may offer advantages in more flexible, nonlinear modeling.

\section{$M$-Forwarding}\label{sec:relay}
We now propose $\msg$-forwarding as a formal definition to embody the concept of how a message $\msg$ is forwarded between two disjoint populations $X_1\in\R^{\nx}$ and $X_2\in\R^{\ny}$, for $\nx$ neurons in $X_1$ and $\ny$ neurons in $X_2$. Three scenarios are possible: $X_1$ forwards $\msg$ to $X_2$, $X_2$ forwards $\msg$ to $X_1$, or neither (i.e., no forwarding occurs or both forward to each other). Inspired by the measure for $M$-Information Flow on an edge as positive conditional mutual information~\cite{venkatesh2020information}, we distinguish between the first two scenarios by formalizing what it means for one area to forward a message to another:
\begin{defi}[$M$-forwarding\footnote{Definition \ref{def:m_comm} is the same as that of a physically degraded broadcast channel.}]\label{def:m_comm}
    For $M\in\R$, $C_1\in \R^{n_1}$, and $C_2\in \R^{n_2}$, where $n_1, n_2 \in \mathbb{N}$, we say $C_1$ $M$-forwards to $C_2$ if the Markov chain $M - C_1 - C_2$ holds, i.e., $I(M; C_2|C_1) = 0$.
\end{defi}
We emphasize that Definition \ref{def:m_comm} requires \emph{all} the information $C_2$ has about $\msg$ must be available in $C_1$ in order to convey the intuitive idea of message forwarding\footnote{We note that this interpretation is not accurate if $C_2$ is an invertible and deterministic function of $C_1$, but is an unlikely scenario in neural data.}. We note that $M$-forward is a notion of communication that is built on a dependence measure, unlike previous works\cite{gokcen2022disentangling, ding2006granger}, whose notions of communication rest on correlational measures.

\section{Experimental results}\label{sec:experiments}
We apply our proposed methodology to data we collected to investigate activity of two neural populations in the whisker-based sensory detection system in mice. We show how IR, dPCA, mTDR, and DSAE applied to each population reveal correlations across the trial. Then, we infer the absence of $\msg$-forwarding on the dimension-reduced populations. We provide additional comparisons to PCA and PLS in Appendix \ref{app:c}.

\subsection{Experiment design and data collection}\label{sec:experiments_a}
We consider a simplified network of whisker-based sensory detection consisting of the primary somatosensory cortex (whisker S1) and the superior colliculus (SC). Whisker S1 is a specialized neural population in the mouse cortex that processes tactile stimuli from whisker deflections~\cite{staiger2021neuronal}. SC is a sub-cortical deep-brain region that is also involved in sensory processing~\cite{hong2018sensation}.  Anatomical evidence has shown that SC receives direct, bottom-up input from the brainstem as well as top-down modulation from S1~\cite{aronoff2010long}. The former is a direct transmission of $\msg$ to SC and happens within 15ms of whisker stimulation, while whisker tactile information likely reaches SC from S1 in the 15-40ms window after stimulation\cite{castro2016whisker}. However, it is unclear whether SC $\msg$-forwards to S1 in this interval~\cite{gharaei2020superior}. The network is shown in Fig. \ref{fig:trial_struct} on the left.

\begin{figure}[htbp]
  \centering
  \includegraphics[width=0.2\textwidth]{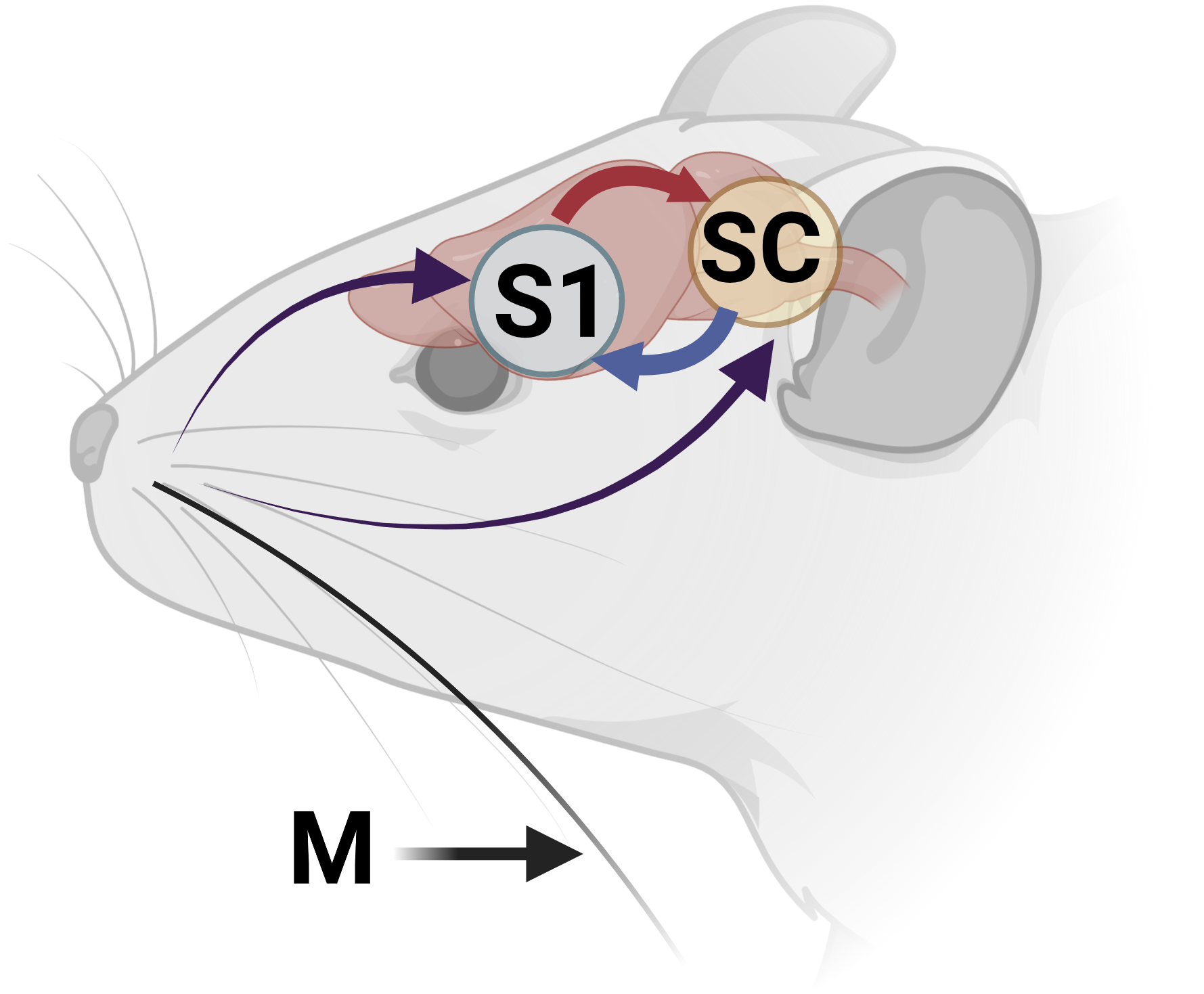}
  \includegraphics[width=0.28\textwidth]{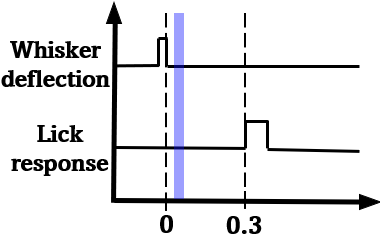}\vspace{-0.1in}
  \caption{(Left) One whisker is deflected while population activity in S1 and SC is recorded. We infer the existence of the red and blue arrows from S1 to SC and SC to S1, respectively. (Right) Trial structure. The whisker is deflected at 0s and the response window opens at 0.3s. The blue shaded region depicts the window between 15-40ms after deflection in which communication is posited to occur, and where we test for $\msg$-forwarding.}
  \label{fig:trial_struct}
\end{figure}
We conducted a pilot experiment on one mouse to understand how S1 and SC interact with respect to activation from whisker deflections. The message $\msg$ is a stimulus taking discrete values that represent varying angles of whisker deflection designed to invoke near-certain detection (large deflection angle), ambiguous detection (medium to small deflection angle), and near-certain absence of detection (zero deflection). The mouse is trained to respond as accurately as possible by licking left or right, and receives a reward for correct detection. Each trial consists of two key events: whisker deflection, followed by licking response (Fig. \ref{fig:trial_struct} right). We investigate the possibility of low-dimensional $\msg$-forwarding in the 15-40ms interval after deflection, as evidence indicates that $\msg$ may be forwarded from S1 to SC within this period~\cite{castro2016whisker}. Appendix \ref{app:b} provides further details about the experiment and data analysis.

\subsection{Marginal $M$-relevance in S1 and SC}\label{sec:pca_ir_1pop}
Fig.~\ref{fig:s1_comparison} illustrates how each method captures $\msg$-relevance in S1. IR, dPCA, and mTDR capture similar $\msg$-relevant activity in their first dimension. DSAE is the only method that is not message-dependent, and this is visibly apparent in smaller and noisier correlations in each dimension. As in S1, IR captures $\msg$-relevant activity in SC in its top dimension, highlighted in three clear peaks (Fig.~\ref{fig:sc_comparison}). Peak 1 is recognized as neural activity responsive to stimulation~\cite{castro2016whisker}. Peaks 2 and 3 likely represent arousal and motor activity related to licking~\cite{drager1975physiology}. The first mTDR dimension identifies a similarly strong peak 2, but diminished peak 1 relative to IR. Peak 3 is evident only in later dimensions (it does not appear as a distinguishable feature in the first dimension). Peaks 1 and 2 are also present in dPCA's first dimension, but are noisier than in IR. Besides peak 1, DSAE shows almost no $\msg$-relevance in SC.

Due to the heterogeneous activity of neurons in SC~\cite{drager1975physiology}, low-dimensional correlations in SC exhibit higher variation relative to S1. While whisker S1 neurons are primarily activated by whisker deflection only, SC is more functionally diverse. $\msg$-relevant activity in SC is not isolated to whisker stimulation~\cite{drager1975physiology, lee2021striatal}, as indicated by the three peaks. IR demonstrates robustness to this noise by maintaining clear $\msg$-relevance in its first dimension. In Appendix~\ref{app:c}, we further show that IR reveals more nuanced activity than the peristimulus time histogram (PSTH), which is another method commonly used to visualize relevant population activity.

\begin{figure}[htbp]
  \centering
  \includegraphics[width=0.45\textwidth]{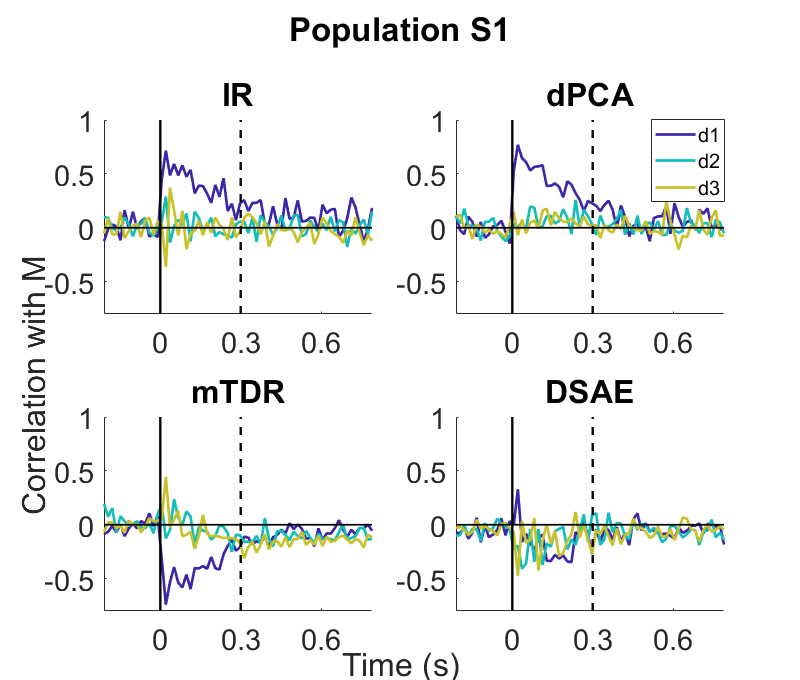}
  \vspace{-0.1in}
  \caption{$\msg$-correlations of low-dimensional representations of S1 for the top three dimensions, d1, d2, and d3. The solid vertical line at 0s marks the time of whisker deflection and the dashed line at 0.3s marks response start.}
  \label{fig:s1_comparison}
\end{figure}

\subsection{Inferring low-dimensional dependencies}\label{sec:infer_m_comm}
We present an application of IR for inferring the existence of a direct pathway between $\msg$ and low-dimensional representations of S1 and SC. As discussed in Sec.~\ref{sec:experiments_a}, we restrict our attention to the 15-40ms window after deflection. First, we identify the top IR dimension for the S1 dataset, $X_{S1}$, and the top IR dimension for the SC dataset, $X_{SC}$. We project onto these linear dimensions, which yields the single `conceptual neuron' $A$ of S1 and $B$ of SC that represent population-level activity most correlated to $\msg$ in the relevant window. Then, we test for dependence between $\msg$ and $A$, as well as $\msg$ and $B$ to verify that the dimension-reduced data still retain relevance with the message~\cite{kraskov2004estimating}. Afterwards, we perform two tests of conditional dependence~\cite{runge2018conditional} for both possible $M$-forward structures: $\msg - A - B$ and $\msg - B - A$. This analysis is repeated for dPCA, mTDR, and DSAE.

Because the null hypothesis of both tests is (conditional) \textit{independence}, rejecting the null of both independence and conditional independence implies the existence of a \emph{direct} pathway from $\msg$. For example, for the data to suggest a direct pathway from $\msg$ to $A$, we would need to reject the null hypotheses $A \indep M$ and $A \indep M \;|\; B$. In summary, rejection indicates that $\msg$ is \emph{not} forwarded from the other conceptual neuron. The data are binned into 15ms intervals, requiring two hypothesis tests to cover the window between 15-40ms after stimulation. 
We test at a level of $\alpha = 0.05$ and apply Bonferroni correction on the estimated p-values. Table \ref{tab:cit_summary} summarizes the results from each hypothesis test, where the top row lists the null hypothesis. $A$ is the dimension-reduced representation of S1 and $B$ is the dimension-reduced representation of SC. `S' means the p-value is significant (i.e., we reject the null hypothesis) and (conditional) dependence can be concluded, while `NS' means not significant and the test is inconclusive. The p-values of each test are reported in Table \ref{tab:p_values} in Appendix \ref{app:c}.

\begin{table}[htbp]
  \caption{Hypothesis tests in 100ms window after deflection}
  \label{tab:cit_summary}
    \centering
    \begin{tabular}{|c|c|c|c|c|}
        \hline
        Method & $A\indep M$ & $B\indep M$ & $A\indep M | B$ & $B \indep M | A$ \\
        \hline\hline
        IR & \emph{S} & \emph{S} & \emph{S} & NS
        \\\hline
        dPCA & \emph{S} & \emph{S} & \emph{S} & NS
        \\\hline
        mTDR & \emph{S} & NS & \emph{S} & NS
        \\\hline
        SAE & \emph{S} & NS & \emph{S} & NS
        \\\hline
    \end{tabular}
\end{table}

\begin{figure}[htbp]
  \centering
  \includegraphics[width=0.43\textwidth]{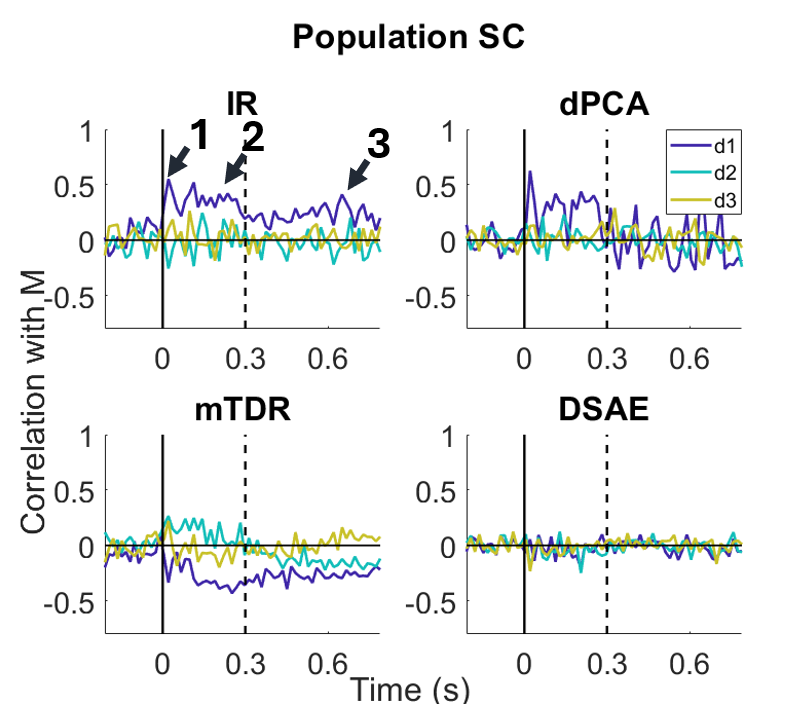}\vspace{-0.1in}
  \caption{$\msg$-correlations of low-dimensional representations of SC for the top three dimensions. Arrows 1, 2, and 3 in IR mark three distinct peaks identified by IR.}
  \label{fig:sc_comparison}
\end{figure}

The hypothesis tests using data pre-processed with IR allow us to confirm that its reduced representations of S1 and SC indeed remain unconditionally dependent with $\msg$ (consistent with visual observation from Fig.~\ref{fig:sc_comparison}). Furthermore, there exists a direct pathway from $\msg$ to $A$. Although this conclusion is made on a low-dimensional representation of S1, it is consistent with existing biological evidence~\cite{hong2018sensation}. There is not enough evidence for the test $B\indep M | A$ to be conclusive, suggesting the possibility of weak or nonexistent $\msg$-forwarding from SC to S1 in this window relative to the evidence collected for S1. Hypothesis tests with dPCA reveal the same structure as IR, but estimate a larger p-value than for IR (Appendix~\ref{app:c}). mTDR and DSAE also detect (un)conditional dependence of S1 with $\msg$, but are unable to conclude that SC depends on $\msg$, even unconditionally.



\section{Discussion}
We propose a new framework for inferring low-dimensional communication between two neural populations that identifies $\msg$-forwarding structures in IR-reduced dimensions. IR may be viewed as a pre-processing step which effectively defines noise relative to the message. Neural activity found to be irrelevant to $\msg$ is removed, even if that activity is part of another biological process. IR is one approach to $\msg$-relevant dimension reduction that is useful when the message is known and measurable, a condition often satisfied in many experimental designs. However, IR is not suitable for settings in which the message is ambiguous, even when the question is one of interaction. For example, following traumatic brain injuries, pathological areas of silence can be `communicated' to other brain areas, where the definition of the message becomes unclear\cite{chamanzar2021neural}.



The alignment of IR conclusions with previous biological results suggests its treatment of noise is effective for inferring an accurate summary of communication. This highlights the importance of choosing an appropriate measure of relevance to address the scientific question.
In principle, our framework could employ other measures of $\msg$-relevance, such as mutual information or distance covariance~\cite{cowley2017distance}, and could extend to nonlinear models with kernel forms. Here, we compute $\msg$-forwarding using only one dimension. Future work will investigate $\msg$-forwarding computed over a larger set of dimensions.

\section*{Acknowledgment}
We thank Aswin Sankaranarayanan for helpful discussions.

\newpage{}

\bibliographystyle{IEEEtran}
\bibliography{Bibliography}

\appendices
\newpage
\section{} \label{app:a}
\noindent
\textbf{Regression is LMMSE and correlation-maximizing vector.}
\\
We show that the vector that maximizes sample correlation also achieves the LMMSE estimate. 
\begin{itemize}
    \item Let $\hat{M} = u^TX$ be a linear approximation of $M$ from $X$ for $u \in\R^n$, where we assume $M$ and $X$ are zero mean. Then, the LMMSE vector is given by
    \begin{align}
        v &= \argmin_u \expv[(\hat{M} - M)^2] \nonumber
        \\
        &= \argmin_{u} \expv[(u^TX - M)^2]. \label{eq:lmmse}
    \end{align}
    The solution to~\eqref{eq:lmmse} is $v = \Sigma_X^{-1}\Sigma_{XM}$.
    \item Let $r(u^TX, M)$ be the sample correlation coefficient between a linear projection of $X$ and $M$. The optimization from~\eqref{eq:max_corr} is re-expressed as
    \begin{align}
        v &= \argmax_u r(u^TX, M) \nonumber
        \\
        &= \argmax_u \frac{u^TXM}{\sqrt{\sigma_m^2} \|u^TX\|}.\label{eq:corr_opt_full}
    \end{align}
    Perform a change of variable by letting $c = X^Tu$. Then,~\eqref{eq:corr_opt_full} becomes
    \begin{align}
        v &= \argmax_c \frac{c^TM}{\|c\|}. \label{eq:corr_var_change}
    \end{align}
    The problem in~\eqref{eq:corr_var_change} is optimized when $c = M$. Using this in our change of variable expression, $v = (XX^T)^{-1}XM = \Sigma_X^{-1}\Sigma_{XM}$.
\end{itemize}
Thus, IR can also be interpreted as finding the best average linear approximation to $\msg$ at each iteration. Theorem~\ref{thm:sufficiency} provides an additional interpretation of IR vectors as the sufficient representation of $\msg$ when $X$ and $M$ are jointly Gaussian.

\begin{theorem}\label{thm:sufficiency}
    Suppose $M \in \R$ and $X \in \R^n$, and $P(M, X) = \mathcal{N}(0, \Sigma)$, where $X$ has marginal covariance $\Sigma_X$. If $\Sigma_X$ is full rank and $v = \gamma \Sigma_{X}^{-1}\Sigma_{XM}$ for some constant $\gamma \neq 0$, then $X$ projected onto $v$ is a sufficient representation of $M$. In other words, the Markov chain $M - v^TX - X$ holds.
\end{theorem}
\begin{proof}
    We first define the relevant distributions: $P(M, X)$, $P(M, X, v^TX)$, and $P(M, X | v^TX)$. Then, we derive the condition for sufficiency and show that the condition holds for $v=\gamma \Sigma_X^{-1}\Sigma_{XM}$, i.e. $I(M; X|v^TX) = 0$. 
    \\
    For $P(M, X)$, define its covariance matrix as
    \begin{align*}
        \Sigma = 
        \begin{bmatrix}
            \sigma_M^2 & \Sigma_{MX}
            \\
            \Sigma_{XM} & \Sigma_X
        \end{bmatrix}
    \end{align*}
    where $\Sigma_X \in \R^{n\times n}$ is the marginal covariance matrix of $X$ and $\Sigma_{XM}\in\R^n$ is the cross covariance vector of $X$ and $M$. Now, let $P(M, X, v^TX) = \mathcal{N}(\mu_0, \Sigma_0)$. The covariance $\Sigma_0$ is
    \begin{align*}
        \Sigma_0 = 
        \begin{bmatrix}
            \sigma_M^2 & \Sigma_{MX} & \Sigma_{MX}v
            \\
            \Sigma_{XM} & \Sigma_X & \Sigma_X v
            \\
            v^T\Sigma_{XM} & v^T\Sigma_X & v^T\Sigma_Xv
        \end{bmatrix}.
    \end{align*}
    Finally, we derive the conditional distribution $P(M, X|v^TX)$. Using properties of multivariate Gaussian distributions, the covariance $\Sigma_C$ of the conditional distribution is
    \begin{align*}
        \Sigma_C &= 
        \begin{bmatrix}
            \sigma_m^2 & \Sigma_{MX}
            \\
            \Sigma_{XM} & \Sigma_X
        \end{bmatrix}
        -
        \begin{bmatrix}
            \frac{\Sigma_{MX}v}{v^T\Sigma_Xv} \\ \frac{\Sigma_X v}{v^T\Sigma_Xv}
        \end{bmatrix}
        \begin{bmatrix}
            v^T\Sigma_{XM} & v^T\Sigma_X
        \end{bmatrix}
        \\
        &= 
        \begin{bmatrix}
            \sigma_m^2 - \frac{\Sigma_{MX}vv^T\Sigma_{XM}}{v^T\Sigma_Xv} & \Sigma_{MX} - \frac{\Sigma_{MX}vv^T\Sigma_X}{v^T\Sigma_Xv}
            \\
            \Sigma_{XM} - \frac{\Sigma_Xvv^T\Sigma_{XM}}{v^T\Sigma_Xv} & \Sigma_X - \frac{\Sigma_X vv^T \Sigma_X}{v^T\Sigma_X v}
        \end{bmatrix}.
    \end{align*}
    Because zero correlation implies independence in jointly Gaussian distributions, setting the off-diagonals of $\Sigma_C$ to zero implies $I(M; X|v^TX) = 0$. Without loss of generality, we derive the condition for sufficiency using the top right matrix. This yields the condition:
    \begin{gather}
        \Sigma_{MX} - \frac{\Sigma_{MX} vv^T \Sigma_X}{v^T\Sigma_Xv} = 0 \nonumber
        \\
        \Leftrightarrow v^T\Sigma_Xv \Sigma_{MX} = \Sigma_{MX}v v^T\Sigma_X. \label{eq:suff_condition}
    \end{gather}
    When $v = \gamma \Sigma_X^{-1}\Sigma_{XM}$, the LHS of~\eqref{eq:suff_condition} equals the RHS:
    \begin{align*}
        v^T\Sigma_Xv \Sigma_{MX} &= \gamma^2 \Sigma_{MX}\Sigma_X^{-1} \Sigma_X \Sigma_X^{-1}\Sigma_{XM}\Sigma_{MX} 
        \\
        &= \gamma^2 \Sigma_{MX}\Sigma_X^{-1}\Sigma_{XM} \Sigma_{MX}
        \\
        &= \gamma^2 \Sigma_{MX} \Sigma_X^{-1}\Sigma_{XM} \Sigma_{MX}\Sigma_X^{-1}\Sigma_X
        \\
        &= \Sigma_{MX}v v^T\Sigma_X.
    \end{align*}
\end{proof}

\noindent
\textbf{Proof of Theorem \ref{thm:pca_ir}.}
\begin{proof}
    The correlation coefficient is
    \begin{align*}
        \corr(p_i^T\dmat, \msg) = \frac{
            \sum_{j=1}^T \langle p_i, x_j  \rangle \msg_j
        }{
        \bar{s}_\msg \sqrt{\sum_{j = 1}^T \langle p_i, x_j \rangle^2}
        },
    \end{align*}
    where $x_j$ is the $j$th column of $\dmat$ and
    \begin{align*}
        \bar{s}_\msg = \sqrt{\sum_{j = 1}^T M_j^2}.
    \end{align*}
    The correlation coefficient is expressed in matrix form as
    \begin{align*}
        \corr(p_i^T\dmat, \msg) = \frac{p_i^T\dmat\msg}{\bar{s}_\msg  \|p_i^T\dmat\|}.
    \end{align*}
    The regression vector is $v = (\dmat\dmat^T)^{-1}\dmat\msg$. Hence, 
    \begin{align*}
        \dmat\msg = \dmat\dmat^T v = (N-1) \Sigma_\dmat v,
    \end{align*}
    where $\Sigma_\dmat$ is the sample covariance matrix, expressed in terms of $\dmat$ as
    \begin{align*}
        \Sigma_\dmat = \frac{\dmat\dmat^T}{N-1}.
    \end{align*}
    Note that the denominator of the correlation coefficient is $\|p_i^T\dmat\| = \sqrt{(p_i^T\dmat)(p_i^T\dmat)^T} = \sqrt{(N-1)p_i^T \Sigma_\dmat p_i}$. Making the appropriate substitutions for $\dmat M$ and $\|p_i^T\dmat\|$, the correlation coefficient can be simplified as
    \begin{align*}
        \corr(p_i^T\dmat, \msg) &= \frac{(N-1)p_i^T \Sigma_\dmat v}{\bar{s}_\msg  \sqrt{(N-1) p_i^T \Sigma_\dmat p_i}}
        \\
        &= \frac{(N-1) \lambda_i \langle v, p_i \rangle}{\bar{s}_\msg  \sqrt{(N-1)\lambda_i \langle p_i, p_i \rangle}}
        \\
        &= \frac{\sqrt{\lambda_i(N-1)}}{\bar{s}_\msg } \langle v, p_i \rangle,
    \end{align*}
    where the second equality is due to $p_i$ being an eigenvector of $\Sigma_\dmat$ with singular value $\lambda_i$. The first IR dimension is $v^* = v / \|v\|$. Projecting with respect to $v^*$ instead of $v$, the expression simplifies as
    \begin{align*}
        \corr(p_i^T\dmat, \msg) &= \alpha \sqrt{\lambda_i} \ \langle v^*, p_i \rangle,
    \end{align*}
    where $\alpha = \|v\| / s_\msg$ and $s_\msg$ is the standard deviation of $\msg$.
\end{proof}

\section{Experiment Methodology \& Data Analysis}\label{app:b}
\noindent
\textbf{Experiment methodology.}
The study protocols were approved by the CMU Institutional Animal Care and Use Committee. Neural spiking activity in one mouse was recorded at 1kHz in a sub-sampled population of neurons in S1 ($n = 38$) and SC ($n=35$) over a total of 208 trials. Each trial is a 5s recording in which the whisker is deflected (i.e., stimulation). The mice are trained to respond after a delay period of 300ms following stimulation to indicate whether it detected a deflection (right lick) or no deflection (left lick) for a water reward. In this paper, we consider the 1s interval between 200ms pre-stimulation and 800ms post-stimulation.

The whisker deflection values are from the set $\{0, 3, 4, 6, 10\}$. The value 0 means no deflection, 10 means max deflection, and the remaining values are designed to induce ambiguous detection ability. Deflection values $0, 3, 4, 6$ and $10$ are randomly presented with probabilities 0.4, 0.1, 0.1, 0.2, and 0.2, respectively. This is to include enough rewarded (i.e. unambiguous) trials to keep the mouse engaged for the duration of the experiment. Figure \ref{fig:trial_struct} (left) was created with BioRender.com.

\noindent
\textbf{Data analysis: pre-processing.}
Unresponsive neurons are first removed, where a neuron is considered unresponsive if it produces less than 15 spikes in the 100ms period after whisker stimulation over all trials with deflection value greater than 0. A total of 31 neurons in S1 and 25 neurons in SC remain after removing unresponsive neurons. The trials are binned into 15ms time bins. We perform 4-fold cross validation so each fold tests on 25\% of the dataset. Figs. \ref{fig:s1_comparison}, \ref{fig:sc_comparison}, and \ref{fig:pca_pls_comparison} show average correlations found on the test data over all folds. 

\noindent
\textbf{Data analysis: model parameters.}
We train mTDR over all time points as this is its model specification. For the remaining methods, we train a new model at each time point. Methods with additional special parameters are summarized below:
\begin{itemize}
    \item \emph{IR}. The regression vectors are found using the statsmodels OLS function in Python~\cite{seabold2010statsmodels}.
    \item \emph{dPCA}. The model is trained in Python with stimulus and time parameters, and without regularization. Three components are identified and their weights from the stimulus-only set are extracted.
    \item \emph{mTDR}. The optimal rank (i.e. number of dimensions) is found by mTDR's greedy estimation procedure. This finds 6 dimensions in S1 and 4 dimensions in SC. The top 3 representations are extracted. Two covariates are used to model the one-dimensional stimulus $\msg$ and a bias term. 
    \item \emph{DSAE}. The model closely follows the network given in~\cite{qiu2018denoising}. The number of units in the input layer is the number of responsive neurons, which varies between S1 and SC. There is only one hidden layer, which consists of three units. All activation functions are sigmoid functions. Our loss function is the mean-square error with a KL divergence loss which encourages sparsity. The KL sparsity parameter is set to $\rho=0.1$. The penalty weight on the KL loss is set to $\beta=0.01$. We apply a corruption operation to achieve a denoising autoencoder. The noise distribution is Poisson with rate parameter set for each input unit to the mean firing rate of the neuron in a given time bin. The noise ratio is set to 10\%, meaning 10\% of the input nodes are corrupted by Poisson noise. The model is trained over 300 epochs with a learning rate of 0.01. As there is no natural ordering of the units, we choose to order them by total absolute correlation with $\msg$ over the whole trial.
\end{itemize}

\noindent
\textbf{Data analysis: hypothesis tests.}
Both independence and conditional independence tests are based on projections obtained by training over the entire dataset in the relevant time window between 15-40ms after stimulation. For mTDR, we train the model over the whole trial before picking out projections from the relevant window. We use the pycit implementation (based on~\cite{yang2019model}) of k-nearest neighbors dependence and conditional dependence testing, drawing 10,000 samples from each population and using 5 neighbors.

\section{Additional comparisons}\label{app:c}
\noindent
\textbf{Peristimulus time histogram (PSTH).} The PSTH is a method commonly applied to spiking data for visualizing key events over the trial~\cite{shimazaki2007method}. A PSTH is computed by summing the total number of times a neuron, or population of neurons, spikes in a given time bin. Fig.~\ref{fig:psth} shows the PSTH computed over all trials and all neurons in S1 and SC. Aside from peak C in the PSTH of SC, no other peaks are immediately evident, in contrast to peaks 1, 2, and 3 in IR (Fig.~\ref{fig:sc_comparison}). This suggests that correlations in IR-projected data can reveal $\msg$-dependent neural activity that may not be evident in the PSTH. Another feature in Fig.~\ref{fig:psth} is peak B in S1. This is not evident in any of the dimension reduction methods in Fig.~\ref{fig:s1_comparison}. Rather, the correlations of the dimension-reduced data suggest that the entire period between the first and second peaks is highly correlated with the message. However, this may also highlight a limitation in using correlation as the measure of $\msg$-relevance, which is limited to capturing linear dependencies.
\begin{figure}[htbp]
  \centering
  \includegraphics[width=0.5\textwidth]{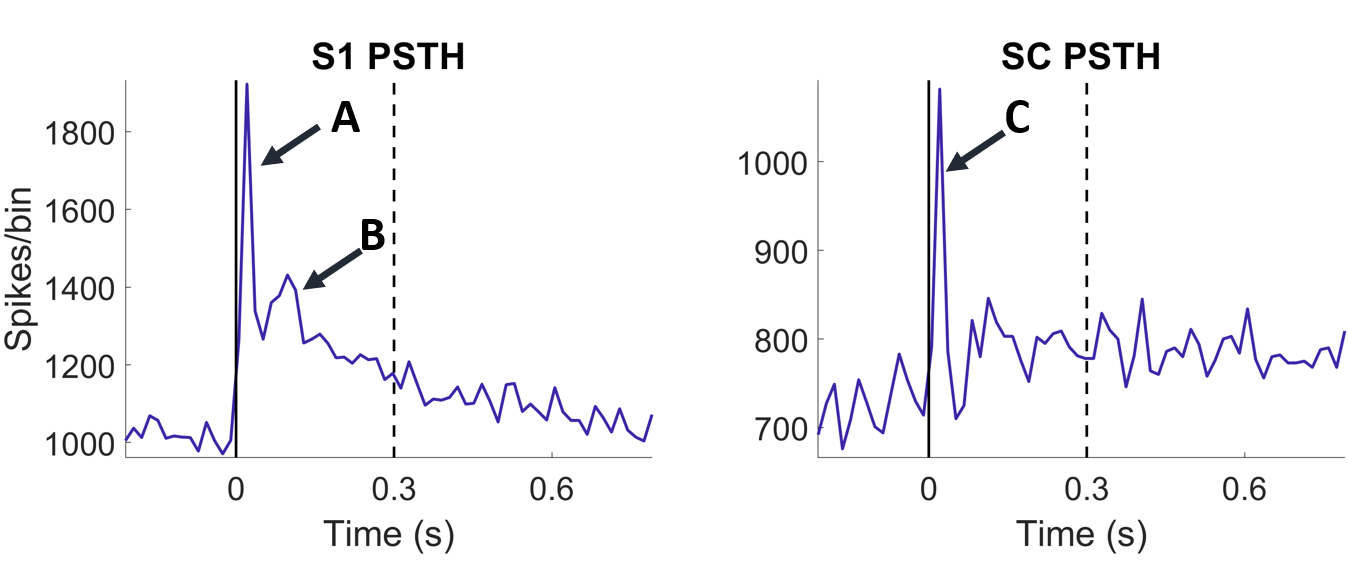}
  \vspace{-0.1in}
  \caption{PSTH for S1 (left) and SC (right) over all trials. The S1 PSTH exhibits two distinctive peaks, labeled A and B, while the SC PSTH only has one evident peak, labeled C.}
  \label{fig:psth}
\end{figure}

\noindent
\textbf{PCA and PLS comparison.}
We provide additional comparisons with PCA and PLS. In S1, the first dimension of both PCA and PLS identify a strong, decaying correlation with $\msg$, and minimal correlations in the other two dimensions. PLS identifies a longer period of $\msg$-relevance than PCA. In SC, PLS identifies the three peaks of $\msg$-relevance similarly to IR in Fig.~\ref{fig:sc_comparison}, albeit noisier. The only visibly identifiable feature in PCA is peak 1 at 0s. The remaining dimensions of PCA only detect small correlations.

\begin{figure}[htbp]
  \centering
  \includegraphics[width=0.45\textwidth]{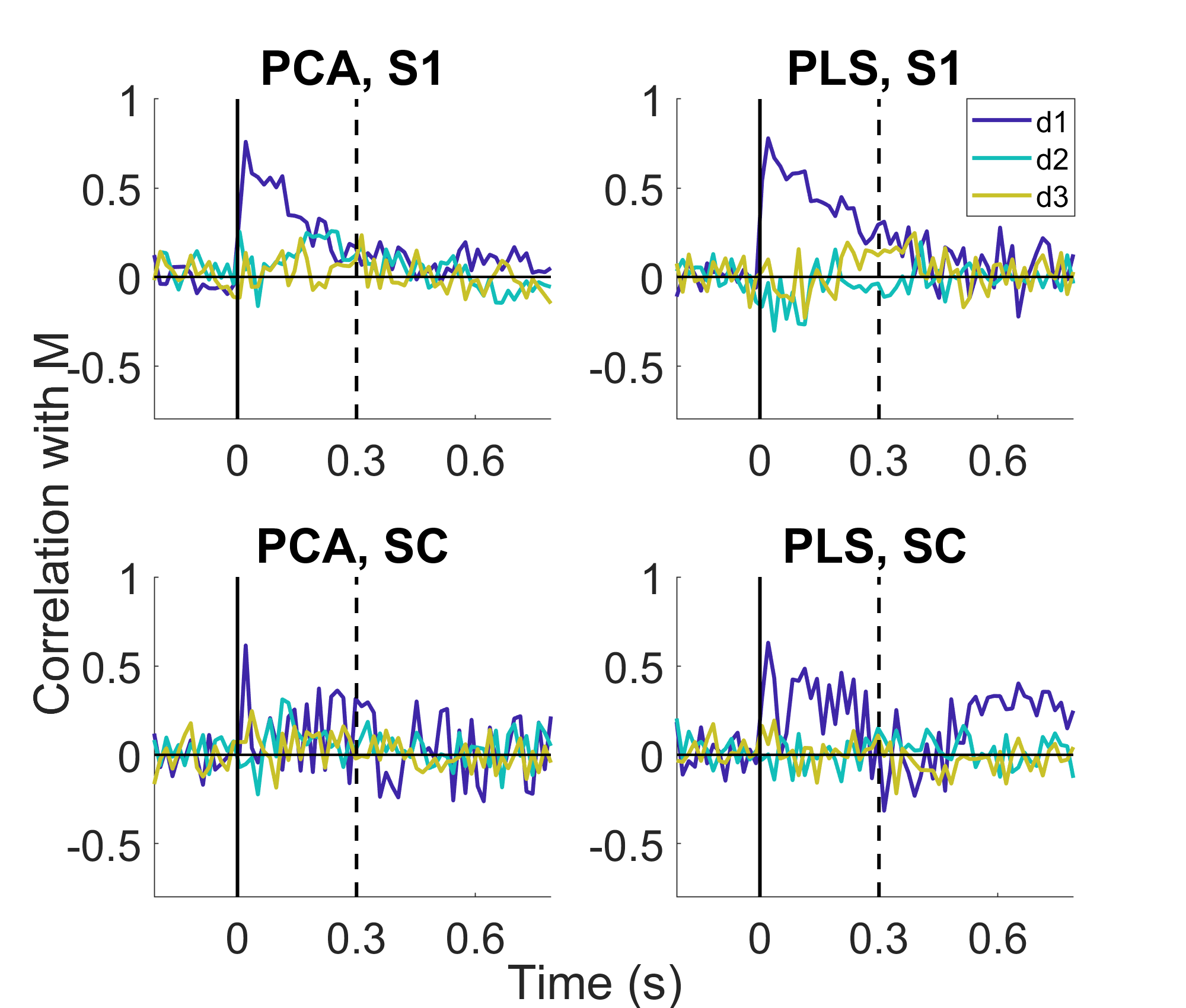}
  \vspace{-0.1in}
  \caption{$\msg$-correlations of low-dimensional representations of S1 and SC using PCA and PLS.}
  \label{fig:pca_pls_comparison}
\end{figure}


\noindent
\textbf{Table of p-values.}
\begin{table}[htbp]
    \caption{p-values for hypothesis tests}
    \label{tab:p_values}
    \centering
    \begin{tabular}{|c|c|c|c|c|c|}
    \hline
        Method & $t$ & {A$\indep\msg$} & {B$\indep\msg$} & {A$\indep\msg \;|\; $B} & {B$\indep\msg \;|\; $A} \\
        \hline\hline
        \multirow{2}{*}{IR} 
        & $p_1$ & $< 10^{-4}$ & $< 10^{-4}$ & $< 10^{-4}$ & $0.0027$
        \\
        & $p_2$ & $< 10^{-4}$ & $< 10^{-4}$ & $< 10^{-4}$ & $0.0364$
        \\\hline
        \multirow{2}{*}{dPCA} 
        & $p_1$ & $< 10^{-4}$ & $< 10^{-4}$ & $< 10^{-4}$ & $0.0374$
        \\
        & $p_2$ & $< 10^{-4}$ & $0.0023$ & $< 10^{-4}$ & $< 10^{-4}$
        \\\hline
        \multirow{2}{*}{mTDR} 
        & $p_1$ & $< 10^{-4}$ & $0.0004$ & $< 10^{-4}$ & $0.0101$
        \\
        & $p_2$ & $< 10^{-4}$ & $0.4376$ & $< 10^{-4}$ & $0.9801$
        \\\hline
        \multirow{2}{*}{DSAE} 
        & $p_1$ & $< 10^{-4}$ & $< 10^{-4}$ & $< 10^{-4}$ & $0.1551$
        \\
        & $p_2$ & $10^{-4}$ & $0.1723$ & $< 10^{-4}$ & $0.0004$
        \\\hline
        \multirow{2}{*}{PCA} 
        & $p_1$ & $< 10^{-4}$ & $< 10^{-4}$ & $< 10^{-4}$ & $0.0159$
        \\
        & $p_2$ & $< 10^{-4}$ & $0.7223$ & $< 10^{-4}$ & $0.5678$
        \\\hline
        \multirow{2}{*}{PLS} 
        & $p_1$ & $< 10^{-4}$ & $< 10^{-4}$ & $< 10^{-4}$ & $0.0023$
        \\
        & $p_2$ & $< 10^{-4}$ & $< 10^{-4}$ & $< 10^{-4}$ & $< 10^{-4}$
        \\\hline
    \end{tabular}
\end{table}

Table \ref{tab:p_values} summarizes the p-values for the two time bins covering the 15-40ms interval after stimulation. Column 2 of the top row indicates the time bin and columns 3-6 list the null hypothesis. $A$ and $B$ are the dimension-reduced representations of S1 and SC, respectively. Since we test at a level of $\alpha=0.05$ with Bonferroni correction, we reject the null hypothesis if each p-value is less than 0.025.

PCA identifies a direct pathway from $\msg$ to S1, but is inconclusive regarding dependence of SC and $\msg$. PLS is the only method that attains significance in all four hypothesis tests. This result means the low-dimensional PLS representation identifies a direct pathway from $\msg$ to S1 and $\msg$ to SC. While this functional structure is biologically plausible, direct connectivity from $\msg$ to SC has been found to occur within 15ms of the stimulus~\cite{castro2016whisker}, which is outside of the tested window.


\IEEEtriggeratref{4}

\end{document}